\documentclass[11pt]{article}%
\usepackage{amsfonts}
\usepackage{amsmath}
\usepackage{amssymb}
\usepackage{fullpage}
\usepackage{hyperref}
\usepackage{graphicx}%
\providecommand{\U}[1]{\protect\rule{.1in}{.1in}}
\newtheorem{theorem}{Theorem}

\newtheorem{lemma}[theorem]{Lemma}

\newtheorem{remark}[theorem]{Remark}

\newenvironment{proof}[1][Proof]{\noindent\textbf{#1.} }{\ \rule{0.5em}{0.5em}}
\begin{document}

\title{\textbf{Squashed entanglement and approximate private states}}
\author{Mark M. Wilde\thanks{Hearne Institute for Theoretical Physics, Department of
Physics and Astronomy, Center for Computation and Technology, Louisiana State
University, Baton Rouge, Louisiana 70803, USA}}
\date{\today}
\maketitle

\begin{abstract}
The squashed entanglement is a fundamental entanglement measure in quantum
information theory, finding application as an upper bound on the distillable
secret key or distillable entanglement of a quantum state or a quantum
channel. This paper simplifies proofs that the squashed entanglement is an
upper bound on distillable key for finite-dimensional quantum systems and
solidifies such proofs for infinite-dimensional quantum systems. More
specifically, this paper establishes that the logarithm of the dimension of
the key system (call it $\log_{2}K$) in an $\varepsilon$-approximate private
state is bounded from above by the squashed entanglement of that state plus a
term that depends only $\varepsilon$ and $\log_{2}K$. Importantly, the extra
term does not depend on the dimension of the shield systems of the private
state. The result holds for the bipartite squashed entanglement, and an
extension of this result is established for two different flavors of the
multipartite squashed entanglement.

\end{abstract}

%\keywords{squashed entanglement, private states, distillable secret key}

\section{Introduction}

The squashed entanglement has become one of the most widely studied
entanglement measures in quantum information theory, due in part to the fact
that it satisfies many of the desirable properties that researchers have
proposed should hold for an entanglement measure \cite{HHHH09}. It was
originally defined in \cite{CW04}\ and shown there to satisfy monotonicity
with respect to local operations and classical communication (LOCC),
convexity, additivity, and reduction to the entanglement entropy for pure
states. Independently, some discussions of a related definition appeared in
\cite{Tucc99,T02}. Later, several different authors proved that squashed
entanglement is asymptotically continuous \cite{AF04}, monogamous \cite{KW04},
and faithful \cite{BCY11}. Multipartite generalizations of squashed
entanglement were independently defined and explored in \cite{AHS08}\ and
\cite{YHHHOS09}, a variety of other information measures related to squashed
entanglement have been presented \cite{YHW08,SBW14,SW14}, and a detailed
investigation of squashed entanglement in infinite-dimensional quantum systems
appeared in \cite{Shir16}. In spite of all of the properties that squashed
entanglement possesses, it is not known whether the quantity is computable in
the Turing sense.

One of the most valuable properties that squashed entanglement possesses is
that it is an upper bound on the distillable entanglement of a bipartite state
\cite{CW04}. This result was later strengthened in
\cite{Chr06,CEHHOR07,Christandl2012}:\ squashed entanglement is also an upper
bound on the distillable secret key of a bipartite state. These results were
further strengthened in \cite{TGW14IEEE}, where the squashed entanglement of a
quantum communication channel was defined and shown to be an upper bound on
the secret key agreement capacity of a quantum channel (i.e., the maximum rate
at which secret key can be distilled by two parties connected by a quantum
channel and free public classical communication links). Multipartite
generalizations of these results are available in \cite{YHHHOS09,STW16}.

The original proof that the squashed entanglement is an upper bound on the
distillable key of a bipartite state $\rho_{AB}$ contained a rather slight
ambiguity \cite[Proposition~4.19]{Chr06}, which was later clarified in
\cite{CEHHOR07,Christandl2012}. At first glance, the issue might appear to be
somewhat technical, but it is in fact critical for having a complete proof of
this result. It is worthwhile to point out that no such issue exists in
various proofs that the relative entropy of entanglement is an upper bound on
distillable key \cite{HHHO05,HHHO09,WTB16}, due to the proof of
\cite[Theorem~9]{HHHO09} and related bounds.

To spell out the issue in more detail, consider that the goal of any key
distillation protocol is for two parties (Alice $A$ and Bob $B$) to act on $n$
independent copies of a shared bipartite state $\rho_{AB}$ using LOCC in order
to distill a so-called private state \cite{HHHO05,HHHO09}, which consists of
two components: key systems and shield systems. Alice and Bob's distilled key
is placed in the key systems, and the shield systems are extra systems
inaccessible to any third eavesdropping party (Eve) who possesses a purifying
system of $\rho_{AB}^{\otimes n}$ and can keep a local copy of all classical
communication exchanged between Alice and Bob during the protocol. The shield
systems are not in the possession of Eve, their purpose being to protect the
key systems from Eve. However, in such a general protocol for key
distillation, the dimension of the shield systems can be arbitrarily large.
This aspect of the protocol is what led to a slight ambiguity in the proof
from \cite[Proposition~4.19]{Chr06}, wherein a parameter $d$ is stated, but it
is left unclear as to whether this is equal to the dimension of the key
systems or the dimension of the key and shield systems combined. Interpreting
the proof there, the only option seems to be that $d$ is equal to the
dimension of the combined key and shield systems, in which case the proof
given in \cite[Proposition~4.19]{Chr06} does not generally establish that
squashed entanglement bounds distillable key from above (i.e., there could
exist a sequence of key distillation protocols resulting in shield systems
with a dimension growing larger than an exponential in $n$, and in such a case
the proof does not establish squashed entanglement as an upper bound on
distillable key). This ambiguity was later resolved in
\cite{CEHHOR07,Christandl2012} for finite-dimensional quantum states, by
noting that all such sequences of protocols can be simulated by ones in which
the shield systems are growing no larger than an exponential in $n$. This
latter argument resolves the aforementioned problem for key distillation
protocols operating on finite-dimensional quantum states, but there is still a
gap left open for such protocols operating on infinite-dimensional quantum
states, since the shield systems in this latter context are inherently
infinite-dimensional. At the same time, it seems desirable at a fundamental
level for the proof to hold regardless of the dimension of the shield systems
(i.e., without the need for a simulation argument).

The present paper settles this issue, which has the simultaneous effect of
1)\ simplifying the proof that the squashed entanglement of a
finite-dimensional state or channel is an upper bound on its distillable key
and 2) solidifying the proof that the same is true for an infinite-dimensional
state or channel. In particular, one of the main results of this paper is that
the logarithm of the dimension of one key system (call it $\log_{2}K$) of an
$\varepsilon$-approximate private state is bounded from above by its squashed
entanglement plus a term that depends only $\varepsilon$ and $\log_{2}K$. The
important point here is that the upper bound has no dependence on the
dimension of the shield systems of the $\varepsilon$-approximate private
state. See Theorem~\ref{thm:bipartite-bound} for a precise statement of the
result. With this new result in hand, we provide a brief review of the proof
that squashed entanglement is an upper bound on distillable key. This paper
also delivers similar results for multipartite squashed entanglements (see
Theorems~\ref{thm:multipartite-bound-1} and \ref{thm:multipartite-bound-2} for
precise statements). The upshot is a full justification of the original
statements from \cite{TGW14IEEE,TGW14Nat,STW16} and the follow-up statements
in \cite{GEW16,AML16,AK16}, regarding distillation of secret key using bosonic
quantum Gaussian channels.

In the next section, we review some preliminary material needed to understand
the main results of the paper. After that, we proceed to establishing proofs
of the main results: Theorems~\ref{thm:bipartite-bound},
\ref{thm:multipartite-bound-1}, and~\ref{thm:multipartite-bound-2}.

\section{Preliminaries}

Much of the background on quantum information theory reviewed here is
available in \cite{W15book}, with the exception of private states and squashed entanglement.

\subsection{Quantum states}

Let $\mathcal{L}(\mathcal{H})$ denote the algebra of bounded linear operators
acting on a Hilbert space $\mathcal{H}$. Let $\mathcal{L}_{+}(\mathcal{H})$
denote the subset of positive semi-definite operators. An\ operator $\rho$ is
in the set $\mathcal{D}(\mathcal{H})$\ of density operators (or states) if
$\rho\in\mathcal{L}_{+}(\mathcal{H})$ and Tr$\left\{  \rho\right\}  =1$. The
tensor product of two Hilbert spaces $\mathcal{H}_{A}$ and $\mathcal{H}_{B}$
is denoted by $\mathcal{H}_{A}\otimes\mathcal{H}_{B}$ or $\mathcal{H}_{AB}%
$.\ Given a multipartite density operator $\rho_{AB}\in\mathcal{D}%
(\mathcal{H}_{A}\otimes\mathcal{H}_{B})$, we unambiguously write $\rho
_{A}=\operatorname{Tr}_{B}\{\rho_{AB}\}$ for the reduced density operator on
system $A$. We use $\rho_{AB}$, $\sigma_{AB}$, $\tau_{AB}$, $\omega_{AB}$,
etc.~to denote general density operators in $\mathcal{D}(\mathcal{H}%
_{A}\otimes\mathcal{H}_{B})$, while $\psi_{AB}$, $\varphi_{AB}$, $\phi_{AB}$,
etc.~denote rank-one density operators (pure states) in $\mathcal{D}%
(\mathcal{H}_{A}\otimes\mathcal{H}_{B})$ (with it implicit, clear from the
context, and the above convention implying that $\psi_{A}$, $\varphi_{A}$,
$\phi_{A}$ may be mixed if $\psi_{AB}$, $\varphi_{AB}$, $\phi_{AB}$ are pure).
A purification $|\phi^{\rho}\rangle_{RA}\in\mathcal{H}_{R}\otimes
\mathcal{H}_{A}$ of a state $\rho_{A}\in\mathcal{D}(\mathcal{H}_{A})$ is such
that $\rho_{A}=\operatorname{Tr}_{R}\{|\phi^{\rho}\rangle\langle\phi^{\rho
}|_{RA}\}$. As is conventional, we often say that a unit vector $|\psi\rangle$
is a pure state or a pure-state vector (while also saying that $|\psi
\rangle\langle\psi|$ is a pure state). An extension of a state $\rho_{A}%
\in\mathcal{S}(  \mathcal{H}_{A})  $ is some state $\rho_{RA}%
\in\mathcal{S}(  \mathcal{H}_{R}\otimes\mathcal{H}_{A})  $ such
that $\operatorname{Tr}_{R}\{  \rho_{RA}\}  =\rho_{A}$. Often, an
identity operator is implicit if we do not write it explicitly (and it should
be clear from the context).

Let $\{|i\rangle_{A}\}$ denote the standard, orthonormal basis for a Hilbert
space $\mathcal{H}_{A}$, and let $\{|i\rangle_{B}\}$ be defined similarly for
$\mathcal{H}_{B}$. If these spaces are finite-dimensional and their dimensions
are equal ($\dim(\mathcal{H}_{A})=\dim(\mathcal{H}_{B})=K$), then we define
the maximally entangled state $|\Phi\rangle_{AB}\in\mathcal{H}_{A}%
\otimes\mathcal{H}_{B}$ as%
\begin{equation}
|\Phi\rangle_{AB}\equiv\frac{1}{\sqrt{K}}\sum_{i}|i\rangle_{A}\otimes
|i\rangle_{B}.
\end{equation}

\subsection{Trace distance and fidelity}

The trace distance between two quantum states $\rho,\sigma\in\mathcal{D}%
(\mathcal{H})$\ is equal to $\left\Vert \rho-\sigma\right\Vert _{1}$, where
$\left\Vert C\right\Vert _{1}\equiv\operatorname{Tr}\{\sqrt{C^{\dag}C}\}$ for
any operator $C$. It has a direct operational interpretation in terms of the
distinguishability of these states. That is, if $\rho$ or $\sigma$ are
prepared with equal probability and the task is to distinguish them via some
quantum measurement, then the optimal success probability in doing so is equal
to $\left(  1+\left\Vert \rho-\sigma\right\Vert _{1}/2\right)  /2$ \cite{H69}.

The fidelity is defined as $F(\rho,\sigma)\equiv\left\Vert \sqrt{\rho}%
\sqrt{\sigma}\right\Vert _{1}^{2}$ \cite{U76}. Uhlmann's theorem states that
\cite{U76}%
\begin{equation}
F(\rho_{A},\sigma_{A})=\max_{U}\left\vert \langle\phi^{\sigma}|_{RA}%
U_{R}\otimes I_{A}|\phi^{\rho}\rangle_{RA}\right\vert ^{2},
\label{eq:uhlmann-thm}%
\end{equation}
where $|\phi^{\rho}\rangle_{RA}$ and $|\phi^{\sigma}\rangle_{RA}$ are fixed
purifications of $\rho_{A}$ and $\sigma_{A}$, respectively, and the
optimization is with respect to all unitaries $U_{R}$. Uhlmann's theorem also
implies that, for a given extension of $\rho_{AB}$ of $\rho_{A}$, there exists
an extension $\sigma_{AB}$ of $\sigma_{A}$ such that%
\begin{equation}
F(\rho_{A},\sigma_{A})=F(\rho_{AB},\sigma_{AB}). \label{eq:uhlmann-extend}%
\end{equation}
See, e.g., \cite[Corollary~3.1]{T15} for an explicit proof of the above
equality. The following inequalities hold for trace distance and fidelity
\cite{FG98}:%
\begin{equation}
1-\sqrt{F(\rho,\sigma)}\leq\frac{1}{2}\left\Vert \rho-\sigma\right\Vert
_{1}\leq\sqrt{1-F(\rho,\sigma)}. \label{eq:fid-trace}%
\end{equation}

\subsection{Private states}

\label{sec:private-states}

Let $\gamma_{ABA^{\prime}B^{\prime}}\in\mathcal{D}(\mathcal{H}_{AA^{\prime
}BB^{\prime}})$\ be a state shared between spatially separated parties Alice
and Bob, such that $K\equiv\dim(\mathcal{H}_{A})=\dim(\mathcal{H}_{B}%
)<+\infty$, Alice possesses systems $A$ and $A^{\prime}$, and Bob possesses
systems $B$ and $B^{\prime}$. The state $\gamma_{ABA^{\prime}B^{\prime}}$ is called a
private state \cite{HHHO05,HHHO09} if Alice and Bob can extract a secret key
from it by performing local measurements on $A$ and $B$, which is product with
any purifying system of $\gamma_{ABA^{\prime}B^{\prime}}$. That is,
$\gamma_{ABA^{\prime}B^{\prime}}$ is a private state of $\log_{2}K$ private
bits if, for any purification $\left\vert \varphi^{\gamma}\right\rangle
_{ABA^{\prime}B^{\prime}E}$ of $\gamma_{ABA^{\prime}B^{\prime}}$, the
following holds:%
\begin{equation}
\left(  \mathcal{M}_{A}\otimes\mathcal{M}_{B}\otimes\text{Tr}_{A^{\prime
}B^{\prime}}\right)  \left(  \varphi_{ABA^{\prime}B^{\prime}E}^{\gamma
}\right)  =\frac{1}{K}\sum_{i}|i\rangle\langle i|_{A}\otimes|i\rangle\langle
i|_{B}\otimes\sigma_{E},
\end{equation}
where $\mathcal{M}(\cdot)=\sum_{i}|i\rangle\langle i|(\cdot)|i\rangle\langle
i|$ is a projective measurement channel and $\sigma_{E}$ is some state on the
purifying system $E$ (which could depend on the particular purification). The
systems $A^{\prime}$ and $B^{\prime}$ are known as \textquotedblleft shield
systems\textquotedblright\ because they aid in keeping the key secure from any
party possessing the purifying system (part or all of which might belong to a
malicious party). It is a non-trivial consequence of the above definition that
a private state of $\log_{2}K$ private bits can be written in the following
form \cite{HHHO05,HHHO09}:
\begin{equation}
\gamma_{ABA^{\prime}B^{\prime}}=U_{ABA^{\prime}B^{\prime}}\left(  \Phi
_{AB}\otimes\sigma_{A^{\prime}B^{\prime}}\right)  U_{ABA^{\prime}B^{\prime}%
}^{\dag},\label{eq:private-1}
\end{equation}
where $\Phi_{AB}$ is a maximally entangled state of Schmidt rank $K$%
\begin{equation}
\Phi_{AB}\equiv\frac{1}{K}\sum_{i,j}|i\rangle\langle j|_{A}\otimes
|i\rangle\langle j|_{B},\label{eq:max-ent-state}%
\end{equation}
and%
\begin{equation}
U_{ABA^{\prime}B^{\prime}}=\sum_{i,j}|i\rangle\langle i|_{A}\otimes
|j\rangle\langle j|_{B}\otimes U_{A^{\prime}B^{\prime}}^{ij}%
\end{equation}
is a controlled unitary known as a \textquotedblleft twisting
unitary,\textquotedblright\ with each $U_{A^{\prime}B^{\prime}}^{ij}$ a
unitary operator.\ Any extension $\gamma_{AA^{\prime}BB^{\prime}E}%
\in\mathcal{D}(\mathcal{H}_{AA^{\prime}BB^{\prime}E})$\ of a private state
$\gamma_{AA^{\prime}BB^{\prime}}$ necessarily has the following form:%
\begin{equation}
\gamma_{AA^{\prime}BB^{\prime}E}=U_{AA^{\prime}BB^{\prime}}\left(  \Phi
_{AB}\otimes\sigma_{A^{\prime}B^{\prime}E}\right)  U_{AA^{\prime}BB^{\prime}%
}^{\dag},\label{eq:ext-private-state}%
\end{equation}
where $\sigma_{A^{\prime}B^{\prime}E}$ is an extension of $\sigma_{A^{\prime
}B^{\prime}}$.

A multipartite private state is a straightforward generalization of the
bipartite definition \cite{HA06}. Indeed, $\gamma_{A_{1}\cdots A_{m}%
A_{1}^{\prime}\cdots A_{m}^{\prime}}$\ is a state of $\log_{2}K$ private bits
if, for any purification $\left\vert \varphi^{\gamma}\right\rangle
_{A_{1}\cdots A_{m}A_{1}^{\prime}\cdots A_{m}^{\prime}E}$ of $\gamma
_{A_{1}\cdots A_{m}A_{1}^{\prime}\cdots A_{m}^{\prime}}$, the following
holds:
\begin{equation}
\left(  \mathcal{M}_{A_{1}}\otimes\cdots\otimes\mathcal{M}_{A_{m}}%
\otimes\text{Tr}_{A_{1}^{\prime}\cdots A_{m}^{\prime}}\right)  \left(
\varphi_{A_{1}\cdots A_{m}A_{1}^{\prime}\cdots A_{m}^{\prime}E}^{\gamma
}\right)  =\frac{1}{K}\sum_{i}|i\rangle\langle i|_{A_{1}}\otimes\cdots
\otimes|i\rangle\langle i|_{A_{m}}\otimes\sigma_{E},\label{eq:mult-secret-key}%
\end{equation}
where $\mathcal{M}$ and $\sigma$ are as before, the key systems $A_{1}$,
\ldots, $A_{m}$ all have the same dimension equal to $K$, and the shield
systems $A_{1}^{\prime}$, \ldots, $A_{m}^{\prime}$ have arbitrary dimension.
The above implies that an $m$-partite private state of $\log_{2}K$ private
bits is a quantum state $\gamma_{A_{1}\cdots A_{m}A_{1}^{\prime}\cdots
A_{m}^{\prime}}$ that can be written as%
\begin{equation}
\gamma_{A_{1}\cdots A_{m}A_{1}^{\prime}\cdots A_{m}^{\prime}}=U_{A_{1}\cdots
A_{m}A_{1}^{\prime}\cdots A_{m}^{\prime}}(\Phi_{A_{1}\cdots A_{m}}%
\otimes\sigma_{A_{1}^{\prime}\cdots A_{m}^{\prime}})U_{A_{1}\cdots A_{m}%
A_{1}^{\prime}\cdots A_{m}^{\prime}}^{\dag},\label{eq:mult_private state}%
\end{equation}
where $\Phi_{A_{1}\cdots A_{m}}$ is an $m$-qudit maximally entangled (GHZ)
state%
\begin{equation}
\Phi_{A_{1}\cdots A_{m}}\equiv\frac{1}{K}\sum_{i,j}|i\rangle\langle j|_{A_{1}%
}\otimes\cdots\otimes|i\rangle\langle j|_{A_{m}}%
\end{equation}
and%
\begin{equation}
U_{A_{1}\cdots A_{m}A_{1}^{\prime}\cdots A_{m}^{\prime}}=\sum_{i_{1}%
,\ldots,i_{m}}|i_{1},\ldots,i_{m}\rangle\langle i_{1},\ldots,i_{m}%
|_{A_{1}\cdots A_{m}}\otimes U_{A_{1}^{\prime}\cdots A_{m}^{\prime}}%
^{i_{1},\ldots,i_{m}}%
\end{equation}
is a twisting unitary, where each unitary $U_{A_{1}^{\prime}\cdots
A_{m}^{\prime}}^{i_{1},\ldots,i_{m}}$ depends on the values $i_{1}%
,\ldots,i_{m}$. Any extension $\gamma_{A_{1}\cdots A_{m}A_{1}^{\prime}\cdots
A_{m}^{\prime}E}$\ of such a private state necessarily has the following form:%
\begin{equation}
\gamma_{A_{1}\cdots A_{m}A_{1}^{\prime}\cdots A_{m}^{\prime}E}=U_{A_{1}\cdots
A_{m}A_{1}^{\prime}\cdots A_{m}^{\prime}}(\Phi_{A_{1}\cdots A_{m}}%
\otimes\sigma_{A_{1}^{\prime}\cdots A_{m}^{\prime}E})U_{A_{1}\cdots A_{m}%
A_{1}^{\prime}\cdots A_{m}^{\prime}}^{\dag},
\end{equation}
where $\sigma_{A_{1}^{\prime}\cdots A_{m}^{\prime}E}$ is an extension of
$\sigma_{A_{1}^{\prime}\cdots A_{m}^{\prime}}$.

\subsection{Conditional quantum mutual and multipartite information}

For a quantum state $\rho_{ABE}$ shared between three parties (Alice, Bob, and
Eve), the conditional quantum mutual information is defined as%
\begin{equation}
I(A;B|E)_{\rho}\equiv H(AE)_{\rho}+H(BE)_{\rho}-H(E)_{\rho}-H(ABE)_{\rho},
\label{eq:CMI-definition}%
\end{equation}
where $H(F)_{\sigma}\equiv-\operatorname{Tr}\{\sigma_{F}\log_{2}\sigma_{F}\}$
is the quantum entropy of a state $\sigma_{F}$ on system $F$. The conditional
quantum entropy is defined as%
\begin{equation}
H(A|B)_{\rho}\equiv H(AB)_{\rho}-H(B)_{\rho},
\end{equation}
which allows us to write%
\begin{equation}
I(A;B|E)_{\rho}=H(A|E)_{\rho}-H(A|BE)_{\rho}. \label{eq:cmi-ce}%
\end{equation}
The conditional quantum mutual information is non-negative:%
\begin{equation}
I(A;B|E)_{\rho}\geq0, \label{eq:SSA}%
\end{equation}
which is an entropy inequality known as strong subadditivity
\cite{LR73,LR73PRL}. The following uniform bound for the continuity of
conditional quantum entropy was proven in \cite{Winter15}, by building on
\cite{AF04}:%
\begin{equation}
\left\vert H(A|B)_{\rho}-H(A|B)_{\omega}\right\vert \leq2\varepsilon\log
_{2}\dim(\mathcal{H}_{A})+(1+\varepsilon)h_{2}(\varepsilon/\left[
1+\varepsilon\right]  ), \label{eq:AFW-ineq}%
\end{equation}
for states $\rho_{AB},\omega_{AB}\in\mathcal{D}(\mathcal{H}_{AB})$ such that%
\begin{equation}
\frac{1}{2}\left\Vert \rho_{AB}-\omega_{AB}\right\Vert _{1}\leq\varepsilon
\in\left[  0,1\right]
\end{equation}
and where the binary entropy $h_{2}(\varepsilon)\equiv-\varepsilon\log_{2}\varepsilon
-(1-\varepsilon)\log_{2}(1-\varepsilon)$. The following uniform bound for
continuity of conditional quantum mutual information holds as well \cite{S15}:%
\begin{equation}
\left\vert I(A;B|E)_{\rho}-I(A;B|E)_{\omega}\right\vert \leq2\varepsilon
\log_{2}\min\left\{  \dim(\mathcal{H}_{A}),\dim(\mathcal{H}_{B})\right\}
+2(1+\varepsilon)h_{2}(\varepsilon/\left[  1+\varepsilon\right]  )).
\label{eq:CMI-cont}%
\end{equation}
for states $\rho_{ABE},\omega_{ABE}\in\mathcal{D}(\mathcal{H}_{ABE})$ such
that $\frac{1}{2}\left\Vert \rho_{ABE}-\omega_{ABE}\right\Vert _{1}%
\leq\varepsilon\in\left[  0,1\right]  $. Notice that this inequality is an
improvement over what one would obtain merely by combining \eqref{eq:cmi-ce}
and \eqref{eq:AFW-ineq}.

For an $m+1$-partite quantum state $\rho_{A_{1}\cdots A_{m}E}$, there are at
least two distinct ways to generalize the conditional mutual information:
\begin{align}
I(A_{1};\cdots;A_{m}|E)_{\rho} &  =\sum_{i=1}^{m}H(A_{i}|E)-H(A_{1}\cdots
A_{m}|E)_{\rho},\label{eq:cmmi_I}\\
\widetilde{I}(A_{1};\cdots;A_{m}|E)_{\rho} &  =\sum_{i=1}^{m}H(A_{\left[
m\right]  \backslash\left\{  i\right\}  }|E)_{\rho}-\left(  m-1\right)
H(A_{1}\cdots A_{m}|E)_{\rho}\label{eq:cmmi_S}\\
&  =H(A_{1}\cdots A_{m}|E)_{\rho}-\sum_{i=1}^{m}H(A_{i}|A_{\left[  m\right]
\backslash\left\{  i\right\}  }E)_{\rho},\label{eq:cmmi_S_2}%
\end{align}
where the shorthand $A_{\left[  m\right]  \backslash\left\{  i\right\}  }$
indicates all systems $A_{1}\cdots A_{m}$ except for system $A_{i}$. Both
quantities are non-negative, due to strong subadditivity. The former is the
conditional version of a quantity known as the total correlation
\cite{Watan60} and has been used in a variety of contexts
\cite{PHH08,YHW08,W14}, while the latter is a conditional version of the dual
total correlation \cite{H75,H78}, employed later on in
\cite{CMS02,YHHHOS09,YHW08}. The above two quantities are generally
incomparable, but related by the following formula \cite{YHHHOS09}:
\begin{equation}
I(A_{1};\cdots;A_{m}|E)_{\rho}+\widetilde{I}(A_{1};\cdots;A_{m}|E)_{\rho}%
=\sum_{i=1}^{m}I(A_{i};A_{\left[  m\right]  \backslash\left\{  i\right\}
}|E)_{\rho}.\label{eq:cmmi_dual}%
\end{equation}
For a state $\rho_{BA_{1}A_{2}\cdots A_{m}E}$, the above conditional
multipartite informations obey the following chain rules, respectively
\cite[Section III]{YHHHOS09}:
\begin{align}
I(BA_{1};\cdots;A_{m}|E)_{\rho} &  =I(A_{1};\cdots;A_{m}|BE)_{\rho}+\sum
_{i=2}^{m}I(B;A_{i}|E)_{\rho},\label{eq:I_Chain}\\
\widetilde{I}(BA_{1};A_{2}\cdots;A_{m}|E)_{\rho} &  =\widetilde{I}(A_{1}%
;A_{2};\cdots;A_{m}|BE)_{\rho}+I(B;A_{2}\cdots A_{m}|E)_{\rho}%
.\label{eq:S_Chain}%
\end{align}

\subsection{Squashed entanglements}

The squashed entanglement of a bipartite state $\rho_{AB}$ is defined as%
\begin{equation}
E_{\operatorname{sq}}(A;B)_{\rho}\equiv\frac{1}{2}\inf_{\omega_{ABE}}\left\{
I(A;B|E)_{\omega}:\rho_{AB}=\operatorname{Tr}_{E}\left\{  \omega
_{ABE}\right\}  \right\}  ,
\end{equation}
where the infimum is with respect to all extensions $\omega_{ABE}$ of the
state $\rho_{AB}$ \cite{CW04}. An interpretation of $E_{\operatorname{sq}%
}\left(  A;B\right)  _{\rho}$\ is that it quantifies the correlations present
between Alice and Bob after a third party (often associated to an environment
or eavesdropper) attempts to \textquotedblleft squash down\textquotedblright%
\ their correlations.

There are at least two different multipartite generalizations of the squashed
entanglement \cite{YHHHOS09,AHS08}. For an $m$-partite quantum state
$\rho_{A_{1}\cdots A_{m}}$, the squashed entanglement measures
$E_{\operatorname{sq}}$ and $\widetilde{E}_{\operatorname{sq}}$ are defined as%
\begin{align}
E_{\operatorname{sq}}(A_{1};\cdots;A_{m})_{\rho}  &  \equiv\frac{{1}}{2}%
\inf_{\omega_{A_{1}A_{2}\cdots A_{m}E}}\left\{  I(A_{1};\cdots;A_{m}%
|E)_{\omega}:\operatorname{Tr}_{E}\left\{  \omega_{A_{1}\cdots A_{m}%
E}\right\}  =\rho_{A_{1}\cdots A_{m}}\right\}  ,\label{eq:I_Squashed}\\
\widetilde{E}_{\operatorname{sq}}(A_{1};\cdots;A_{m})_{\rho}  &  \equiv
\frac{{1}}{2}\inf_{\omega_{A_{1}A_{2}\cdots A_{m}E}}\left\{  \widetilde
{I}(A_{1};\cdots;A_{m}|E)_{\omega}:\operatorname{Tr}_{E}\left\{  \omega
_{A_{1}\cdots A_{m}E}\right\}  =\rho_{A_{1}\cdots A_{m}}\right\}  ,
\label{eq:S_Squashed}%
\end{align}
where the infima are taken with respect to all possible extensions
$\omega_{A_{1}\cdots A_{m}E}$ of $\rho_{A_{1}\cdots A_{m}}$, and $I$ and
$\widetilde{I}$ are the conditional quantum multipartite information
quantities given in \eqref{eq:cmmi_I} and \eqref{eq:cmmi_S}, respectively.

\section{Bipartite squashed entanglement and approximate private states}

This section establishes one of this paper's main results
(Theorem~\ref{thm:bipartite-bound}), which is an upper bound on the logarithm
of the dimension of a key system of an $\varepsilon$-approximate private state
in terms of its squashed entanglement, plus another term depending only on
$\varepsilon$ and $\log_{2}K$. We start with the following lemma, which
applies to any extension of a bipartite private state:

\begin{lemma}
\label{lem:log-K-to-info-measures}Let $\gamma_{AA^{\prime}BB^{\prime}}$ be a
bipartite private state and let $\gamma_{AA^{\prime}BB^{\prime}E}$ be an
extension of it, as defined in Section~\ref{sec:private-states}. Then the
following identity holds for any such extension:%
\begin{equation}
2\log_{2}K=I(A;BB^{\prime}|E)_{\gamma}+I(A^{\prime};B|AB^{\prime}E)_{\gamma}.
\label{eq:logK-to-info-measures}%
\end{equation}

\end{lemma}

\begin{proof}
First consider that the following identity holds as a consequence of two
applications of the chain rule for conditional quantum mutual information:%
\begin{align}
I(AA^{\prime};BB^{\prime}|E)_{\gamma}  &  =I(A;BB^{\prime}|E)_{\gamma
}+I(A^{\prime};BB^{\prime}|AE)_{\gamma}\nonumber\\
&  =I(A;BB^{\prime}|E)_{\gamma}+I(A^{\prime};B^{\prime}|AE)_{\gamma
}+I(A^{\prime};B|B^{\prime}AE)_{\gamma}. \label{eq:chain-rule-CMI}%
\end{align}
Combined with the following identity, which holds for an extension
$\gamma_{AA^{\prime}BB^{\prime}E}$ of a private state $\gamma_{AA^{\prime
}BB^{\prime}}$,%
\begin{equation}
I(AA^{\prime};BB^{\prime}|E)_{\gamma}=2\log_{2}K+I(A^{\prime};B^{\prime
}|AE)_{\gamma}, \label{eq:christandl-thesis}%
\end{equation}
we recover the statement in \eqref{eq:logK-to-info-measures}. So it remains to
prove \eqref{eq:christandl-thesis}. This identity is a very slight rewriting
of the last line in the proof of \cite[Proposition~4.19]{Chr06}, and we recall
the proof here. By definition, we have that%
\begin{equation}
I(AA^{\prime};BB^{\prime}|E)_{\gamma}=H(AA^{\prime}E)_{\gamma}+H(BB^{\prime
}E)_{\gamma}-H(E)_{\gamma}-H(AA^{\prime}BB^{\prime}E)_{\gamma}.
\end{equation}
By applying \eqref{eq:max-ent-state}--\eqref{eq:ext-private-state}, we can write $\gamma_{AA^{\prime
}BB^{\prime}E}$ as follows:%
\begin{equation}
\gamma_{AA^{\prime
}BB^{\prime}E} = \frac{1}{K}\sum_{i,j}|i\rangle\langle j|_{A}\otimes|i\rangle\langle
j|_{B}\otimes U_{A^{\prime}B^{\prime}}^{ii}\sigma_{A^{\prime}B^{\prime}%
E}(U_{A^{\prime}B^{\prime}}^{jj})^{\dag}.
\end{equation}
Tracing over system $B$ leads to the following state:%
\begin{equation}
\gamma_{AA^{\prime}B^{\prime}E}=\frac{1}{K}\sum_{i}|i\rangle\langle
i|_{A}\otimes\gamma_{A^{\prime}B^{\prime}E}^{i}, \label{eq:trace-out-B}%
\end{equation}
where%
\begin{equation}
\gamma_{A^{\prime}B^{\prime}E}^{i}\equiv U_{A^{\prime}B^{\prime}}^{ii}%
\sigma_{A^{\prime}B^{\prime}E}(U_{A^{\prime}B^{\prime}}^{ii})^{\dag}.
\end{equation}
Similarly, tracing over system $A$ of $\gamma_{AA^{\prime
}BB^{\prime}E}$ leads to%
\begin{equation}
\gamma_{BA^{\prime}B^{\prime}E}=\frac{1}{K}\sum_{i}|i\rangle\langle
i|_{B}\otimes\gamma_{A^{\prime}B^{\prime}E}^{i}. \label{eq:trace-out-A}%
\end{equation}
So these and the chain rule for conditional entropy imply that%
\begin{equation}
H(AA^{\prime}E)_{\gamma}=H(A)_{\gamma}+H(A^{\prime}E|A)_{\gamma}=\log_{2}
K+H(A^{\prime}E|A)_{\gamma}.
\end{equation}
Similarly, we have that%
\begin{equation}
H(BB^{\prime}E)_{\gamma}=\log_{2} K+H(B^{\prime}E|B)_{\gamma}=\log_{2}
K+H(B^{\prime}E|A)_{\gamma},
\end{equation}
where we have used the symmetries in
\eqref{eq:trace-out-B}--\eqref{eq:trace-out-A}. Since $\gamma_{E}=\gamma
_{E}^{i}$ for all $i$, we find that%
\begin{equation}
H(E)_{\gamma}=\frac{1}{K}\sum_{i}H(E)_{\gamma^{i}}=H(E|A)_{\gamma}.
\end{equation}
Finally, we have that%
\begin{align}
H(AA^{\prime}BB^{\prime}E)_{\gamma}  &  =H(ABA^{\prime}B^{\prime}%
E)_{\Phi\otimes\sigma}=H(AB)_{\Phi}+H(A^{\prime}B^{\prime}E)_{\sigma}\\
&  =\frac{1}{K}\sum_{i}H(A^{\prime}B^{\prime}E)_{\gamma^{i}}=H(A^{\prime
}B^{\prime}E|A)_{\gamma}.
\end{align}
Combining the above, we recover \eqref{eq:christandl-thesis}.
\end{proof}

\bigskip

We can now establish one of the main results of the paper:

\begin{theorem}
\label{thm:bipartite-bound}Let $\gamma_{AA^{\prime}BB^{\prime}}$ be a private
state and let $\omega_{AA^{\prime}BB^{\prime}}$ be an $\varepsilon
$-approximate private state, in the sense that%
\begin{equation}
F(\gamma_{AA^{\prime}BB^{\prime}},\omega_{AA^{\prime}BB^{\prime}}%
)\geq1-\varepsilon
\end{equation}
for $\varepsilon\in\left[  0,1\right]  $. Then%
\begin{equation}
\log_{2}K\leq E_{\operatorname{sq}}(AA^{\prime};BB^{\prime})_{\omega}%
+f_{1}(\sqrt{\varepsilon},K),
\end{equation}
where%
\begin{equation}
f_{1}(\varepsilon,K)\equiv2\varepsilon\log_{2}K+2(1+\varepsilon)h_{2}%
(\varepsilon/\left[  1+\varepsilon\right]  ).
\end{equation}

\end{theorem}

\begin{proof}
By \eqref{eq:uhlmann-extend} and \eqref{eq:fid-trace}, for a given extension
$\omega_{AA^{\prime}BB^{\prime}E}$ of $\omega_{AA^{\prime}BB^{\prime}}$, there
exists an extension $\gamma_{AA^{\prime}BB^{\prime}E}$ of $\gamma_{AA^{\prime
}BB^{\prime}}$ such that%
\begin{equation}
\frac{1}{2}\left\Vert \gamma_{AA^{\prime}BB^{\prime}E}-\omega_{AA^{\prime
}BB^{\prime}E}\right\Vert _{1}\leq\sqrt{\varepsilon}.
\end{equation}
We then find that%
\begin{align}
2\log_{2}K  &  =I(A;BB^{\prime}|E)_{\gamma}+I(A^{\prime};B|AB^{\prime
}E)_{\gamma}\\
&  \leq I(A;BB^{\prime}|E)_{\omega}+I(A^{\prime};B|AB^{\prime}E)_{\omega
}+2f_{1}(\sqrt{\varepsilon},K)\\
&  \leq I(A;BB^{\prime}|E)_{\omega}+I(A^{\prime};B|AB^{\prime}E)_{\omega
}+I(A^{\prime};B^{\prime}|AE)_{\omega}+2f_{1}(\sqrt{\varepsilon},K)\\
&  =I(AA^{\prime};BB^{\prime}|E)_{\omega}+2f_{1}(\sqrt{\varepsilon},K).
\end{align}
The first equality follows from Lemma~\ref{lem:log-K-to-info-measures}. The
first inequality follows from two applications of \eqref{eq:CMI-cont}. The
second inequality follows because $I(A^{\prime};B^{\prime}|AE)_{\omega}\geq0$
(this is strong subadditivity, recalled in \eqref{eq:SSA}). The last equality
is a consequence of the chain rule for conditional mutual information, as used
in \eqref{eq:chain-rule-CMI}. Since the inequality%
\begin{equation}
2\log_{2} K\leq I(AA^{\prime};BB^{\prime}|E)_{\omega}+2f_{1}(\sqrt
{\varepsilon},K)
\end{equation}
holds for any extension of $\omega$, the statement of the theorem follows.
\end{proof}

\bigskip For completeness, we now provide an arguably simpler proof that
squashed entanglement is an upper bound on distillable key. Before doing so,
let us recall the definition of distillable key of a bipartite state
$\rho_{AB}$. An $(n,P,\varepsilon)$ key distillation protocol for $\rho_{AB}$
consists of an LOCC\ channel $\mathcal{L}_{A^{n}B^{n}\rightarrow\hat{A}\hat
{B}A^{\prime}B^{\prime}}$ such that%
\begin{equation}
F(\omega_{\hat{A}\hat{B}A^{\prime}B^{\prime}},\gamma_{\hat{A}\hat{B}A^{\prime
}B^{\prime}})\geq1-\varepsilon\in\left[  0,1\right]  ,
\end{equation}
where%
\begin{equation}
\omega_{\hat{A}\hat{B}A^{\prime}B^{\prime}}\equiv\mathcal{L}_{A^{n}%
B^{n}\rightarrow\hat{A}\hat{B}A^{\prime}B^{\prime}}(\rho_{AB}^{\otimes n}),
\end{equation}
$\gamma_{\hat{A}\hat{B}A^{\prime}B^{\prime}}$ is a private state, and $\left[
\log_{2}\dim(\mathcal{H}_{\hat{A}})\right]  /n=\left[  \log_{2}\dim
(\mathcal{H}_{\hat{B}})\right]  /n\geq P$. A distillable key rate~$P$ is
achievable for $\rho_{AB}$ if for all $\varepsilon\in(0,1)$, $\delta>0$, and
sufficiently large $n$, there exists an $(n,P-\delta,\varepsilon)$ key
distillation protocol for $\rho_{AB}$. The distillable key $P(\rho_{AB})$ is
defined to be the supremum of all distillable key rates. We can then establish
a slightly simpler proof of the following theorem from
\cite{Chr06,CEHHOR07,Christandl2012}, by employing
Theorem~\ref{thm:bipartite-bound} in the first step of the proof:

\begin{theorem}
[\cite{Chr06,CEHHOR07,Christandl2012}]\label{thm:distillable-key-state}The
distillable key $P(\rho_{AB})$ of a bipartite state $\rho_{AB}$ is bounded
from above by its squashed entanglement:%
\begin{equation}
P(\rho_{AB})\leq E_{\operatorname{sq}}(A;B)_{\rho}.
\end{equation}

\end{theorem}

\begin{proof}
Consider an arbitrary $(n,P,\varepsilon)$ key distillation protocol for
$\rho_{AB}$. We then have that%
\begin{align}
\log_{2}\dim(\mathcal{H}_{\hat{A}})  &  \leq E_{\operatorname{sq}}(\hat
{A}A^{\prime};\hat{B}B^{\prime})_{\omega}+f_{1}(\sqrt{\varepsilon},\log
_{2}\dim(\mathcal{H}_{\hat{A}}))\\
&  \leq E_{\operatorname{sq}}(A^{n};B^{n})_{\rho^{\otimes n}}+f_{1}%
(\sqrt{\varepsilon},\log_{2}\dim(\mathcal{H}_{\hat{A}}))\\
&  =nE_{\operatorname{sq}}(A;B)_{\rho}+f_{1}(\sqrt{\varepsilon},\log_{2}%
\dim(\mathcal{H}_{\hat{A}})).
\end{align}
The inequalities follow respectively from Theorem~\ref{thm:bipartite-bound},
LOCC monotonicity of squashed entanglement \cite{CW04}, and additivity of
squashed entanglement with respect to tensor-product states \cite{CW04}. We
can then write the above explicitly as%
\begin{equation}
P\leq\frac{1}{n}\log_{2}\dim(\mathcal{H}_{\hat{A}})\leq\frac{1}{1-2\sqrt
{\varepsilon} }E_{\operatorname{sq}}(A;B)_{\rho}+\frac{2(1+\sqrt{\varepsilon
})}{n(1-2\sqrt{\varepsilon} )}h_{2}(\sqrt{\varepsilon}/\left[  1+\sqrt
{\varepsilon}\right]  ),
\end{equation}
whenever $1-2\sqrt{\varepsilon}>0$.
Taking the limit as $n\rightarrow\infty$ and then as $\varepsilon\rightarrow0$
establishes the result.
\end{proof}

\begin{remark}
\label{rem:distillable-key-channel}An $(n,P,\varepsilon)$ key distillation
protocol which employs a quantum channel $\mathcal{N}$ is defined similarly,
except one allows for $n$ uses of a quantum channel, with each use interleaved
by a round of LOCC (see \cite{TGW14IEEE} for a precise definition). One
defines achievable rates similarly as above, and $P_{2}(\mathcal{N})$ is the
LOCC-assisted private capacity of a quantum channel $\mathcal{N}$, equal to
the supremum of all achievable rates. A similar argument as in the proof of
Theorem~\ref{thm:distillable-key-state}, along with a particular subadditivity
lemma for squashed entanglement from \cite{TGW14IEEE}, can be used to
establish the following bound for an $(n,P,\varepsilon)$ key distillation
protocol which employs a quantum channel $\mathcal{N}$:
\begin{equation}
P\leq\frac{1}{1-2\sqrt{\varepsilon}}E_{\operatorname{sq}}(\mathcal{N}%
)+\frac{2(1+\sqrt{\varepsilon})}{n(1-2\sqrt{\varepsilon})}h_{2}(\sqrt
{\varepsilon}/\left[  1+\sqrt{\varepsilon}\right]  ),\label{eq:channel-bound}%
\end{equation}
whenever $1-2\sqrt{\varepsilon}>0$.
In the above, $E_{\operatorname{sq}}(\mathcal{N})$ is the squashed
entanglement of a quantum channel $\mathcal{N}_{A^{\prime}\rightarrow B}$,
defined in \cite{TGW14IEEE} as%
\begin{align}
E_{\operatorname{sq}}(\mathcal{N})  & \equiv\max_{\psi_{AA^{\prime}}%
}E_{\operatorname{sq}}(A;B)_{\omega},\\
\omega_{AB}  & \equiv\mathcal{N}_{A^{\prime}\rightarrow B}(\psi_{AA^{\prime}%
}),
\end{align}
where the optimization is with respect to all pure states $\psi_{AA^{\prime}}$
with $\dim(\mathcal{H}_{A})=\dim(\mathcal{H}_{A^{\prime}})$.\ The inequality
in \eqref{eq:channel-bound} implies that $P_{2}(\mathcal{N})\leq
E_{\operatorname{sq}}(\mathcal{N})$. See \cite{TGW14IEEE} for further details.
\end{remark}

\section{Multipartite squashed entanglements and approximate private states}

We can handle the multipartite squashed entanglements in a similar way. The
proof strategies are similar, with the main idea being to find particular
representations for the following quantities:%
\begin{align}
& I(A_{1}A_{1}^{\prime};\cdots;A_{m}A_{m}^{\prime}|E)_{\gamma}-I(A_{1}%
^{\prime};\cdots;A_{m}^{\prime}|EA_{1})_{\gamma},\\
& \widetilde{I}(A_{1}A_{1}^{\prime};\cdots;A_{m}A_{m}^{\prime}|E)_{\gamma
}-\widetilde{I}(A_{1}^{\prime};\cdots;A_{m}^{\prime}|EA_{1})_{\gamma},
\end{align}
each of which was previously shown to be equal to $m\log_{2}K$ (see
\cite[Eqs.~(78)--(80)]{YHHHOS09} and \cite[Eqs.~(162)--(164)]{STW16},
respectively). These representations are in terms of information quantities
which can be bounded from above by the dimensions of the key systems, so that
we can employ uniform continuity estimates \cite{Winter15}\ for them in which
the only dimension terms appearing are those of the key systems.

We begin by considering the first multipartite squashed entanglement in \eqref{eq:I_Squashed}.

\begin{lemma}
\label{lem:log-K-to-info-measures-multi-1}Let $\gamma_{A_{1}\cdots A_{m}%
A_{1}^{\prime}\cdots A_{m}^{\prime}}$ be a multipartite private state and let
$\gamma_{A_{1}\cdots A_{m}A_{1}^{\prime}\cdots A_{m}^{\prime}E}$ be an
extension of it, as defined in Section~\ref{sec:private-states}. Then the
following identity holds for any such extension:%
\begin{equation}
m\log_{2} K=\sum_{i=2}^{m}H(A_{i}|A_{i}^{\prime}EA_{1})_{\gamma}+\sum
_{i=2}^{m}I(A_{1};A_{i}A_{i}^{\prime}|E)_{\gamma} -H(A_{2}\cdots A_{m}%
|EA_{1}A_{1}^{\prime}\cdots A_{m}^{\prime})_{\gamma}.
\end{equation}

\end{lemma}

\begin{proof}
The following identity holds for multipartite private states
\cite[Eqs.~(78)--(80)]{YHHHOS09}:%
\begin{equation}
I(A_{1}A_{1}^{\prime};\cdots;A_{m}A_{m}^{\prime}|E)_{\gamma}=m\log_{2}
K+I(A_{1}^{\prime};\cdots;A_{m}^{\prime}|EA_{1})_{\gamma}. \label{eq:yang-id}%
\end{equation}
Now, consider that%
\begin{align}
&  I(A_{1}A_{1}^{\prime};\cdots;A_{m}A_{m}^{\prime}|E)_{\gamma}-I(A_{1}%
^{\prime};\cdots;A_{m}^{\prime}|EA_{1})_{\gamma}\nonumber\\
&  =I(A_{1}^{\prime};A_{2}A_{2}^{\prime};\cdots;A_{m}A_{m}^{\prime}%
|EA_{1})_{\gamma}+\sum_{i=2}^{m}I(A_{1};A_{i}A_{i}^{\prime}|E)_{\gamma
}-I(A_{1}^{\prime};\cdots;A_{m}^{\prime}|EA_{1})_{\gamma}%
\label{eq:multi-step-1}\\
&  =H(A_{1}^{\prime}|EA_{1})_{\gamma}+\sum_{i=2}^{m}H(A_{i}A_{i}^{\prime
}|EA_{1})_{\gamma}-H(A_{1}^{\prime}A_{2}A_{2}^{\prime}\cdots A_{m}%
A_{m}^{\prime}|EA_{1})_{\gamma}\nonumber\\
&  \qquad+\sum_{i=2}^{m}I(A_{1};A_{i}A_{i}^{\prime}|E)_{\gamma}-\left[
H(A_{1}^{\prime}|EA_{1})_{\gamma}+\sum_{i=2}^{m}H(A_{i}^{\prime}%
|EA_{1})_{\gamma}-H(A_{1}^{\prime}\cdots A_{m}^{\prime}|EA_{1})_{\gamma
}\right] \\
&  =\sum_{i=2}^{m}H(A_{i}A_{i}^{\prime}|EA_{1})_{\gamma}-H(A_{1}^{\prime}%
A_{2}A_{2}^{\prime}\cdots A_{m}A_{m}^{\prime}|EA_{1})_{\gamma}+\sum_{i=2}%
^{m}I(A_{1};A_{i}A_{i}^{\prime}|E)_{\gamma}\nonumber\\
&  \qquad-\sum_{i=2}^{m}H(A_{i}^{\prime}|EA_{1})_{\gamma}+H(A_{1}^{\prime
}\cdots A_{m}^{\prime}|EA_{1})_{\gamma}\\
&  =\sum_{i=2}^{m}H(A_{i}|A_{i}^{\prime}EA_{1})_{\gamma}-H(A_{2}\cdots
A_{m}|EA_{1}A_{1}^{\prime}\cdots A_{m}^{\prime})_{\gamma}+\sum_{i=2}%
^{m}I(A_{1};A_{i}A_{i}^{\prime}|E)_{\gamma}. \label{eq:multi-step-last}%
\end{align}
The first equality follows from \eqref{eq:I_Chain}. The second equality
follows by expanding the multipartite information quantities using their
definitions. The last equality follows because%
\begin{align}
H(A_{i}A_{i}^{\prime}|EA_{1})_{\gamma}-H(A_{i}^{\prime}|EA_{1})_{\gamma}  &
=H(A_{i}|A_{i}^{\prime}EA_{1})_{\gamma},\\
-H(A_{1}^{\prime}A_{2}A_{2}^{\prime}\cdots A_{m}A_{m}^{\prime}|EA_{1}%
)_{\gamma}+H(A_{1}^{\prime}\cdots A_{m}^{\prime}|EA_{1})_{\gamma}  &
=-H(A_{2}\cdots A_{m}|EA_{1}A_{1}^{\prime}\cdots A_{m}^{\prime})_{\gamma}.
\end{align}
Putting \eqref{eq:multi-step-1}--\eqref{eq:multi-step-last} together with
\eqref{eq:yang-id} gives the statement of the lemma.
\end{proof}

\begin{theorem}
\label{thm:multipartite-bound-1}Let $\gamma_{A_{1}\cdots A_{m}A_{1}^{\prime
}\cdots A_{m}^{\prime}}$ be a multipartite private state, as defined in
Section~\ref{sec:private-states}, and let $\omega_{A_{1}\cdots A_{m}%
A_{1}^{\prime}\cdots A_{m}^{\prime}}$ be an $\varepsilon$-approximate private
state, in the sense that%
\begin{equation}
F(\gamma_{A_{1}\cdots A_{m}A_{1}^{\prime}\cdots A_{m}^{\prime}},\omega
_{A_{1}\cdots A_{m}A_{1}^{\prime}\cdots A_{m}^{\prime}})\geq1-\varepsilon
\end{equation}
for $\varepsilon\in\left[  0,1\right]  $. Then%
\begin{equation}
\frac{m}{2}\log_{2}K\leq E_{\operatorname{sq}}(A_{1}A_{1}^{\prime}%
;\cdots;A_{m}A_{m}^{\prime})_{\omega}+f_{2}(\sqrt{\varepsilon},K),
\end{equation}
where%
\begin{equation}
f_{2}(\varepsilon,K,m)\equiv m\left[  b_{1}\varepsilon\log_{2}K+b_{2}%
(1+\varepsilon)h_{2}(\varepsilon/\left[  1+\varepsilon\right]  )\right]  ,
\end{equation}
for some constants $b_{1},b_{2}\in\mathbb{Z}^{+}$.
\end{theorem}

\begin{proof}
By \eqref{eq:uhlmann-extend} and \eqref{eq:fid-trace}, for a given extension
$\omega_{A_{1}\cdots A_{m}A_{1}^{\prime}\cdots A_{m}^{\prime}E}$ of
$\omega_{A_{1}\cdots A_{m}A_{1}^{\prime}\cdots A_{m}^{\prime}}$, there exists
an extension $\gamma_{A_{1}\cdots A_{m}A_{1}^{\prime}\cdots A_{m}^{\prime}E}$
of $\gamma_{A_{1}\cdots A_{m}A_{1}^{\prime}\cdots A_{m}^{\prime}}$ such that%
\begin{equation}
\frac{1}{2}\left\Vert \gamma_{A_{1}\cdots A_{m}A_{1}^{\prime}\cdots
A_{m}^{\prime}E}-\omega_{A_{1}\cdots A_{m}A_{1}^{\prime}\cdots A_{m}^{\prime
}E}\right\Vert _{1}\leq\sqrt{\varepsilon}.
\end{equation}
We then find that%
\begin{align}
m\log_{2}K &  =\sum_{i=2}^{m}H(A_{i}|A_{i}^{\prime}EA_{1})_{\gamma}+\sum
_{i=2}^{m}I(A_{1};A_{i}A_{i}^{\prime}|E)_{\gamma}-H(A_{2}\cdots A_{m}%
|EA_{1}A_{1}^{\prime}\cdots A_{m}^{\prime})_{\gamma}\\
&  \leq\sum_{i=2}^{m}H(A_{i}|A_{i}^{\prime}EA_{1})_{\omega}+\sum_{i=2}%
^{m}I(A_{1};A_{i}A_{i}^{\prime}|E)_{\omega}\nonumber\\
&  \qquad-H(A_{2}\cdots A_{m}|EA_{1}A_{1}^{\prime}\cdots A_{m}^{\prime
})_{\omega}+2f_{2}(\sqrt{\varepsilon},K,m)\\
&  \leq\sum_{i=2}^{m}H(A_{i}|A_{i}^{\prime}EA_{1})_{\omega}+\sum_{i=2}%
^{m}I(A_{1};A_{i}A_{i}^{\prime}|E)_{\omega}-H(A_{2}\cdots A_{m}|EA_{1}%
A_{1}^{\prime}\cdots A_{m}^{\prime})_{\omega}\nonumber\\
&  \qquad+I(A_{1}^{\prime};\cdots;A_{m}^{\prime}|EA_{1})_{\omega}+2f_{2}%
(\sqrt{\varepsilon},K,m)\\
&  =I(A_{1}A_{1}^{\prime};\cdots;A_{m}A_{m}^{\prime}|E)_{\omega}+2f_{2}%
(\sqrt{\varepsilon},K,m).
\end{align}
The first equality follows from Lemma~\ref{lem:log-K-to-info-measures-multi-1}%
. The first inequality follows from several applications of
\eqref{eq:AFW-ineq}\ and\ \eqref{eq:CMI-cont}. The second inequality follows
because $I(A_{1}^{\prime};\cdots;A_{m}^{\prime}|EA_{1})_{\omega}\geq0$. The
last equality is a consequence of
\eqref{eq:multi-step-1}--\eqref{eq:multi-step-last}, which clearly apply to an
arbitrary state. Since the inequality%
\begin{equation}
m\log_{2}K\leq I(A_{1}A_{1}^{\prime};\cdots;A_{m}A_{m}^{\prime}|E)_{\omega
}+2f_{2}(\sqrt{\varepsilon},K,m)
\end{equation}
holds for any extension of $\omega$, the statement of the theorem follows.
\end{proof}

\bigskip We now handle the other multipartite squashed entanglement from \eqref{eq:S_Squashed}.

\begin{lemma}
\label{lem:log-K-to-info-measures-multi-2}Let $\gamma_{A_{1}\cdots A_{m}%
A_{1}^{\prime}\cdots A_{m}^{\prime}}$ be a multipartite private state, and let
$\gamma_{A_{1}\cdots A_{m}A_{1}^{\prime}\cdots A_{m}^{\prime}E}$ be an
extension of it, as defined in Section~\ref{sec:private-states}. Then the
following identity holds for any such extension:%
\begin{multline}
m\log_{2}K=H(A_{2}\cdots A_{m}|EA_{1}A_{2}^{\prime}\cdots A_{m}^{\prime
})_{\gamma}-\sum_{i=2}^{m}H(A_{i}|EA_{1}A_{\left[  m\right]  }^{\prime
})_{\gamma}\\
+\sum_{i=2}^{m}I(A_{i}A_{i}^{\prime};A_{\left[  m\right]  \backslash
\{i,1\}}|EA_{1}A_{\left[  m\right]  \backslash\{i\}}^{\prime})_{\gamma
}+I(A_{1};A_{2}A_{2}^{\prime}\cdots A_{m}A_{m}^{\prime}|E)_{\gamma}.
\end{multline}

\end{lemma}

\begin{proof}
The following identity holds for an extension of a private state
\cite[Eqs.~(162)--(164)]{STW16}:%
\begin{equation}
\widetilde{I}(A_{1}A_{1}^{\prime};\cdots;A_{m}A_{m}^{\prime}|E)_{\gamma}%
=m\log_{2}K+\widetilde{I}(A_{1}^{\prime};\cdots;A_{m}^{\prime}|EA_{1}%
)_{\gamma}.
\end{equation}
At the same time, we have that%
\begin{align}
&  \widetilde{I}(A_{1}A_{1}^{\prime};\cdots;A_{m}A_{m}^{\prime}|E)_{\gamma
}-\widetilde{I}(A_{1}^{\prime};\cdots;A_{m}^{\prime}|EA_{1})_{\gamma
}\nonumber\\
&  =\widetilde{I}(A_{1}^{\prime};A_{2}A_{2}^{\prime};\cdots;A_{m}A_{m}%
^{\prime}|EA_{1})_{\gamma}+I(A_{1};A_{2}A_{2}^{\prime}\cdots A_{m}%
A_{m}^{\prime}|E)_{\gamma}-\widetilde{I}(A_{1}^{\prime};\cdots;A_{m}^{\prime
}|EA_{1})_{\gamma}\\
&  =H(A_{1}^{\prime}A_{2}A_{2}^{\prime}\cdots A_{m}A_{m}^{\prime}%
|EA_{1})_{\gamma}-H(A_{1}^{\prime}|EA_{1}A_{2}A_{2}^{\prime}\cdots A_{m}%
A_{m}^{\prime})_{\gamma}\nonumber\\
&  \qquad-\sum_{i=2}^{m}H(A_{i}A_{i}^{\prime}|EA_{1}A_{\left[  m\right]
\backslash\{i,1\}}A_{\left[  m\right]  \backslash\{i\}}^{\prime})_{\gamma
}+I(A_{1};A_{2}A_{2}^{\prime}\cdots A_{m}A_{m}^{\prime}|E)_{\gamma}\nonumber\\
&  \qquad-\left[  H(A_{1}^{\prime}\cdots A_{m}^{\prime}|EA_{1})_{\gamma
}-H(A_{1}^{\prime}|EA_{1}A_{2}^{\prime}\cdots A_{m}^{\prime})_{\gamma}%
-\sum_{i=2}^{m}H(A_{i}^{\prime}|EA_{1}A_{\left[  m\right]  \backslash\left\{
i\right\}  }^{\prime})_{\gamma}\right] \\
&  =H(A_{2}\cdots A_{m}|EA_{1}A_{1}^{\prime}\cdots A_{m}^{\prime})_{\gamma
}+I(A_{1}^{\prime};A_{2}\cdots A_{m}|EA_{1}A_{2}^{\prime}\cdots A_{m}^{\prime
})_{\gamma}\nonumber\\
&  \qquad-\sum_{i=2}^{m}H(A_{i}A_{i}^{\prime}|EA_{1}A_{\left[  m\right]
\backslash\{i,1\}}A_{\left[  m\right]  \backslash\{i\}}^{\prime})_{\gamma
}+I(A_{1};A_{2}A_{2}^{\prime}\cdots A_{m}A_{m}^{\prime}|E)_{\gamma}\nonumber\\
&  \qquad+\sum_{i=2}^{m}H(A_{i}^{\prime}|EA_{1}A_{\left[  m\right]
\backslash\left\{  i\right\}  }^{\prime})_{\gamma}
\label{eq:multi-2-last-line}%
\end{align}
The first equality follows from \eqref{eq:S_Chain}. The second equality
follows by expanding using \eqref{eq:cmmi_S_2}. The third equality follows
because%
\begin{align}
H(A_{1}^{\prime}A_{2}A_{2}^{\prime}\cdots A_{m}A_{m}^{\prime}|EA_{1})_{\gamma
}-H(A_{1}^{\prime}\cdots A_{m}^{\prime}|EA_{1})_{\gamma}  &  =H(A_{2}\cdots
A_{m}|EA_{1}A_{1}^{\prime}\cdots A_{m}^{\prime})_{\gamma},\\
-H(A_{1}^{\prime}|EA_{1}A_{2}A_{2}^{\prime}\cdots A_{m}A_{m}^{\prime}%
)_{\gamma}+H(A_{1}^{\prime}|EA_{1}A_{2}^{\prime}\cdots A_{m}^{\prime}%
)_{\gamma}  &  =I(A_{1}^{\prime};A_{2}\cdots A_{m}|EA_{1}A_{2}^{\prime}\cdots
A_{m}^{\prime})_{\gamma}.
\end{align}
Continuing,%
\begin{align}
\eqref{eq:multi-2-last-line}  &  =H(A_{2}\cdots A_{m}|EA_{1}A_{1}^{\prime
}\cdots A_{m}^{\prime})_{\gamma}+I(A_{1}^{\prime};A_{2}\cdots A_{m}%
|EA_{1}A_{2}^{\prime}\cdots A_{m}^{\prime})_{\gamma}\nonumber\\
&  \qquad-\sum_{i=2}^{m}H(A_{i}A_{i}^{\prime}|EA_{1}A_{\left[  m\right]
\backslash\{i\}}^{\prime})_{\gamma}+\sum_{i=2}^{m}I(A_{i}A_{i}^{\prime
};A_{\left[  m\right]  \backslash\{i,1\}}|EA_{1}A_{\left[  m\right]
\backslash\{i\}}^{\prime})_{\gamma}\nonumber\\
&  \qquad+I(A_{1};A_{2}A_{2}^{\prime}\cdots A_{m}A_{m}^{\prime}|E)_{\gamma
}+\sum_{i=2}^{m}H(A_{i}^{\prime}|EA_{1}A_{\left[  m\right]  \backslash\left\{
i\right\}  }^{\prime})_{\gamma}\\
&  =H(A_{2}\cdots A_{m}|EA_{1}A_{2}^{\prime}\cdots A_{m}^{\prime})_{\gamma
}-\sum_{i=2}^{m}H(A_{i}|EA_{1}A_{\left[  m\right]  }^{\prime})_{\gamma
}\nonumber\\
&  \qquad+\sum_{i=2}^{m}I(A_{i}A_{i}^{\prime};A_{\left[  m\right]
\backslash\{i,1\}}|EA_{1}A_{\left[  m\right]  \backslash\{i\}}^{\prime
})_{\gamma}+I(A_{1};A_{2}A_{2}^{\prime}\cdots A_{m}A_{m}^{\prime}|E)_{\gamma}.
\end{align}
The first equality follows because%
\begin{multline}
-\sum_{i=2}^{m}H(A_{i}A_{i}^{\prime}|EA_{1}A_{\left[  m\right]  \backslash
\{i,1\}}A_{\left[  m\right]  \backslash\{i\}}^{\prime})_{\gamma}=-\sum
_{i=2}^{m}H(A_{i}A_{i}^{\prime}|EA_{1}A_{\left[  m\right]  \backslash
\{i\}}^{\prime})_{\gamma}\\
+\sum_{i=2}^{m}I(A_{i}A_{i}^{\prime};A_{\left[  m\right]  \backslash
\{i,1\}}|EA_{1}A_{\left[  m\right]  \backslash\{i\}}^{\prime})_{\gamma},
\end{multline}
and the second because%
\begin{multline}
H(A_{2}\cdots A_{m}|EA_{1}A_{1}^{\prime}\cdots A_{m}^{\prime})_{\gamma
}+I(A_{1}^{\prime};A_{2}\cdots A_{m}|EA_{1}A_{2}^{\prime}\cdots A_{m}^{\prime
})_{\gamma}\\
=H(A_{2}\cdots A_{m}|EA_{1}A_{2}^{\prime}\cdots A_{m}^{\prime})_{\gamma},
\end{multline}%
\begin{equation}
-\sum_{i=2}^{m}H(A_{i}A_{i}^{\prime}|EA_{1}A_{\left[  m\right]  \backslash
\{i\}}^{\prime})_{\gamma}+\sum_{i=2}^{m}H(A_{i}^{\prime}|EA_{1}A_{\left[
m\right]  \backslash\left\{  i\right\}  }^{\prime})_{\gamma}=-\sum_{i=2}%
^{m}H(A_{i}|EA_{1}A_{\left[  m\right]  }^{\prime})_{\gamma}.
\end{equation}
This concludes the proof.
\end{proof}

\bigskip

We state a final theorem without proof, as it goes similarly to the proof of
Theorem~\ref{thm:multipartite-bound-1}.

\begin{theorem}
\label{thm:multipartite-bound-2}Let $\gamma_{A_{1}\cdots A_{m}A_{1}^{\prime
}\cdots A_{m}^{\prime}}$ be a private state and let $\omega_{A_{1}\cdots
A_{m}A_{1}^{\prime}\cdots A_{m}^{\prime}}$ be an $\varepsilon$-approximate
private state, in the sense that%
\begin{equation}
F(\gamma_{A_{1}\cdots A_{m}A_{1}^{\prime}\cdots A_{m}^{\prime}},\omega
_{A_{1}\cdots A_{m}A_{1}^{\prime}\cdots A_{m}^{\prime}})\geq1-\varepsilon
\end{equation}
for $\varepsilon\in\left[  0,1\right]  $. Then%
\begin{equation}
\frac{m}{2}\log_{2}K\leq\widetilde{E}_{\operatorname{sq}}(A_{1}A_{1}^{\prime
};\cdots;A_{m}A_{m}^{\prime})_{\omega}+f_{3}(\sqrt{\varepsilon},K),
\end{equation}
where%
\begin{equation}
f_{3}(\varepsilon,K,m)\equiv m\left[  c_{1}\varepsilon\log_{2}K+c_{2}%
(1+\varepsilon)h_{2}(\varepsilon/\left[  1+\varepsilon\right]  )\right]  ,
\end{equation}
for some constants $c_{1},c_{2}\in\mathbb{Z}^{+}$.
\end{theorem}

\begin{remark}
Theorems~\ref{thm:multipartite-bound-1} and \ref{thm:multipartite-bound-2} can
be used to establish upper bounds on multipartite distillable key of
multipartite states and broadcast channels, in a way similar to
Theorem~\ref{thm:distillable-key-state} and
Remark~\ref{rem:distillable-key-channel}. See \cite{YHHHOS09} and \cite{STW16}
for details.
\end{remark}

\bigskip

\textbf{Acknowledgements.} I am grateful to Koji Azuma and Stefan Baeuml for
pointing out the main issue discussed in this paper. I am as well thankful to
Koji Azuma, Stefan Baeuml, Saikat Guha, Ryan Gregory James, Masahiro Takeoka, and Stephanie Wehner
for discussions related to the topic of this paper. I thank the anonymous
referees for helpful comments that improved the readability of the paper. I acknowledge support
from the NSF\ under Award No.~CCF-1350397 and thank Stefano Mancini for
hosting me at University of Camerino during late June 2016, where this result
was developed.

\bibliographystyle{alpha}
\bibliography{Ref}

\newcommand{\etalchar}[1]{$^{#1}$}
\begin{thebibliography}{TGW14b}

\bibitem[AF04]{AF04}
Robert Alicki and Mark Fannes.
\newblock Continuity of quantum conditional information.
\newblock {\em Journal of Physics A: Mathematical and General}, 37(5):L55--L57,
  February 2004.
\newblock arXiv:quant-ph/0312081.

\bibitem[AHS08]{AHS08}
David Avis, Patrick Hayden, and Ivan Savov.
\newblock Distributed compression and multiparty squashed entanglement.
\newblock {\em Journal of Physics A: Mathematical and Theoretical},
  41(11):115301, March 2008.
\newblock arXiv:0707.2792.

\bibitem[AK17]{AK16}
Koji Azuma and Go~Kato.
\newblock Aggregating quantum repeaters for the quantum internet.
\newblock {\em Physical Review A}, 96(3):032332, September 2017.
\newblock arXiv:1606.00135.

\bibitem[AML16]{AML16}
Koji Azuma, Akihiro Mizutani, and Hoi-Kwong Lo.
\newblock Fundamental rate-loss tradeoff for the quantum internet.
\newblock {\em Nature Communications}, 7:13523, November 2016.
\newblock arXiv:1601.02933.

\bibitem[BCY11]{BCY11}
Fernando G.~S.~L. Brandao, Matthias Christandl, and Jon Yard.
\newblock Faithful squashed entanglement.
\newblock {\em Communications in Mathematical Physics}, 306(3):805--830,
  September 2011.
\newblock arXiv:1010.1750.

\bibitem[CEH{\etalchar{+}}07]{CEHHOR07}
Matthias Christandl, Artur Ekert, Michal Horodecki, Pawel Horodecki, Jonathan
  Oppenheim, and Renato Renner.
\newblock Unifying classical and quantum key distillation.
\newblock {\em Proceedings of the 4th Theory of Cryptography Conference,
  Lecture Notes in Computer Science}, 4392:456--478, February 2007.
\newblock arXiv:quant-ph/0608199.

\bibitem[Chr06]{Chr06}
Matthias Christandl.
\newblock {\em The Structure of Bipartite Quantum States: Insights from Group
  Theory and Cryptography}.
\newblock PhD thesis, University of Cambridge, April 2006.
\newblock arXiv:quant-ph/0604183.

\bibitem[CMS02]{CMS02}
Nicolas~J. Cerf, Serge Massar, and Sara Schneider.
\newblock Multipartite classical and quantum secrecy monotones.
\newblock {\em Physical Review A}, 66(4):042309, October 2002.
\newblock arXiv:quant-ph/0202103.

\bibitem[CSW12]{Christandl2012}
Matthias Christandl, Norbert Schuch, and Andreas Winter.
\newblock Entanglement of the antisymmetric state.
\newblock {\em Communications in Mathematical Physics}, 311(2):397--422, April
  2012.
\newblock arXiv:0910.4151.

\bibitem[CW04]{CW04}
Matthias Christandl and Andreas Winter.
\newblock Squashed entanglement: An additive entanglement measure.
\newblock {\em Journal of Mathematical Physics}, 45(3):829--840, March 2004.
\newblock arXiv:quant-ph/0308088.

\bibitem[FvdG98]{FG98}
Christopher~A. Fuchs and Jeroen van~de Graaf.
\newblock Cryptographic distinguishability measures for quantum mechanical
  states.
\newblock {\em IEEE Transactions on Information Theory}, 45(4):1216--1227, May
  1998.
\newblock arXiv:quant-ph/9712042.

\bibitem[GEW16]{GEW16}
Kenneth Goodenough, David Elkouss, and Stephanie Wehner.
\newblock Assessing the performance of quantum repeaters for all
  phase-insensitive {Gaussian} bosonic channels.
\newblock {\em New Journal of Physics}, 18(6):063005, June 2016.
\newblock arXiv:1511.08710.

\bibitem[HA06]{HA06}
Pawe\l{} Horodecki and Remigiusz Augusiak.
\newblock Quantum states representing perfectly secure bits are always
  distillable.
\newblock {\em Physical Review A}, 74(1):010302, July 2006.
\newblock arXiv:quant-ph/0602176.

\bibitem[Han75]{H75}
Te~Sun Han.
\newblock Linear dependence structure of the entropy space.
\newblock {\em Information and Control}, 29(4):337--368, December 1975.

\bibitem[Han78]{H78}
Te~Sun Han.
\newblock Nonnegative entropy measures of multivariate symmetric correlations.
\newblock {\em Information and Control}, 36(2):133--156, February 1978.

\bibitem[Hel69]{H69}
Carl~W. Helstrom.
\newblock Quantum detection and estimation theory.
\newblock {\em Journal of Statistical Physics}, 1:231--252, 1969.

\bibitem[HHHH09]{HHHH09}
Ryszard Horodecki, Pawe\l{} Horodecki, Micha\l{} Horodecki, and Karol
  Horodecki.
\newblock {Q}uantum entanglement.
\newblock {\em Reviews of Modern Physics}, 81(2):865--942, June 2009.
\newblock arXiv:quant-ph/0702225v2.

\bibitem[HHHO05]{HHHO05}
Karol Horodecki, Micha\l{} Horodecki, Pawe\l{} Horodecki, and Jonathan
  Oppenheim.
\newblock Secure key from bound entanglement.
\newblock {\em Physical Review Letters}, 94(16):160502, April 2005.
\newblock arXiv:quant-ph/0309110.

\bibitem[HHHO09]{HHHO09}
Karol Horodecki, Michal Horodecki, Pawel Horodecki, and Jonathan Oppenheim.
\newblock General paradigm for distilling classical key from quantum states.
\newblock {\em IEEE Transactions on Information Theory}, 55(4):1898--1929,
  April 2009.
\newblock arXiv:quant-ph/0506189.

\bibitem[KW04]{KW04}
Masato Koashi and Andreas Winter.
\newblock Monogamy of quantum entanglement and other correlations.
\newblock {\em Physical Review A}, 69(2):022309, February 2004.
\newblock arXiv:quant-ph/0310037.

\bibitem[LR73a]{LR73PRL}
Elliott~H. Lieb and Mary~Beth Ruskai.
\newblock A fundamental property of quantum-mechanical entropy.
\newblock {\em Physical Review Letters}, 30(10):434--436, March 1973.

\bibitem[LR73b]{LR73}
Elliott~H. Lieb and Mary~Beth Ruskai.
\newblock Proof of the strong subadditivity of quantum-mechanical entropy.
\newblock {\em Journal of Mathematical Physics}, 14(12):1938--1941, December
  1973.

\bibitem[PHH08]{PHH08}
Marco Piani, Pawel Horodecki, and Ryszard Horodecki.
\newblock No-local-broadcasting theorem for multipartite quantum correlations.
\newblock {\em Physical Review Letters}, 100(9):090502, March 2008.
\newblock arXiv:0707.0848.

\bibitem[SBW15]{SBW14}
Kaushik~P. Seshadreesan, Mario Berta, and Mark~M. Wilde.
\newblock R\'enyi squashed entanglement, discord, and relative entropy
  differences.
\newblock {\em Journal of Physics A: Mathematical and Theoretical},
  48(39):395303, September 2015.
\newblock arXiv:1410.1443.

\bibitem[Shi16]{Shir16}
Maksim~E. Shirokov.
\newblock Squashed entanglement in infinite dimensions.
\newblock {\em Journal of Mathematical Physics}, 57(3):032203, March 2016.
\newblock arXiv:1507.08964.

\bibitem[Shi17]{S15}
Maksim~E. Shirokov.
\newblock Tight continuity bounds for the quantum conditional mutual
  information, for the {Holevo} quantity and for capacities of a channel.
\newblock {\em Journal of Mathematical Physics}, 58(10):102202, October 2017.
\newblock arXiv:1512.09047.

\bibitem[STW16]{STW16}
Kaushik~P. Seshadreesan, Masahiro Takeoka, and Mark~M. Wilde.
\newblock Bounds on entanglement distillation and secret key agreement for
  quantum broadcast channels.
\newblock {\em IEEE Transactions on Information Theory}, 62(5):2849--2866, May
  2016.
\newblock arXiv:1503.08139.

\bibitem[SW15]{SW14}
Kaushik~P. Seshadreesan and Mark~M. Wilde.
\newblock Fidelity of recovery, squashed entanglement, and measurement
  recoverability.
\newblock {\em Physical Review A}, 92(4):042321, October 2015.
\newblock arXiv:1410.1441.

\bibitem[TGW14a]{TGW14Nat}
Masahiro Takeoka, Saikat Guha, and Mark~M. Wilde.
\newblock Fundamental rate-loss tradeoff for optical quantum key distribution.
\newblock {\em Nature Communications}, 5:5235, October 2014.
\newblock arXiv:1504.06390.

\bibitem[TGW14b]{TGW14IEEE}
Masahiro Takeoka, Saikat Guha, and Mark~M. Wilde.
\newblock The squashed entanglement of a quantum channel.
\newblock {\em IEEE Transactions on Information Theory}, 60(8):4987--4998,
  August 2014.
\newblock arXiv:1310.0129.

\bibitem[Tom16]{T15}
Marco Tomamichel.
\newblock {\em Quantum Information Processing with Finite Resources ---
  Mathematical Foundations}, volume~5 of {\em SpringerBriefs in Mathematical
  Physics}.
\newblock Springer, 2016.
\newblock arXiv:1504.00233.

\bibitem[Tuc99]{Tucc99}
Robert~R. Tucci.
\newblock Quantum entanglement and conditional information transmission, 1999.
\newblock arXiv:quant-ph/9909041v2.

\bibitem[Tuc02]{T02}
Robert~R. Tucci.
\newblock Entanglement of distillation and conditional mutual information.
\newblock 2002.
\newblock arXiv:quant-ph/0202144.

\bibitem[Uhl76]{U76}
Armin Uhlmann.
\newblock The ``transition probability'' in the state space of a *-algebra.
\newblock {\em Reports on Mathematical Physics}, 9(2):273--279, 1976.

\bibitem[Wat60]{Watan60}
Satosi Watanabe.
\newblock Information theoretical analysis of multivariate correlation.
\newblock {\em IBM Journal of Research and Development}, 4(1):66--82, January
  1960.

\bibitem[Wil14]{W14}
Mark~M. Wilde.
\newblock Multipartite quantum correlations and local recoverability.
\newblock {\em Proceedings of the Royal Society A}, 471:20140941, March 2014.
\newblock arXiv:1412.0333.

\bibitem[Wil16]{W15book}
Mark~M. Wilde.
\newblock {\em From Classical to Quantum Shannon Theory}.
\newblock March 2016.
\newblock arXiv:1106.1445v7.

\bibitem[Win16]{Winter15}
Andreas Winter.
\newblock Tight uniform continuity bounds for quantum entropies: conditional
  entropy, relative entropy distance and energy constraints.
\newblock {\em Communications in Mathematical Physics}, 347(1):291--313,
  October 2016.
\newblock arXiv:1507.07775.

\bibitem[WTB17]{WTB16}
Mark~M. Wilde, Marco Tomamichel, and Mario Berta.
\newblock Converse bounds for private communication over quantum channels.
\newblock {\em IEEE Transactions on Information Theory}, 63(3):1792--1817,
  March 2017.
\newblock arXiv:1602.08898.

\bibitem[YHH{\etalchar{+}}09]{YHHHOS09}
Dong Yang, Karol Horodecki, Michal Horodecki, Pawel Horodecki, Jonathan
  Oppenheim, and Wei Song.
\newblock Squashed entanglement for multipartite states and entanglement
  measures based on the mixed convex roof.
\newblock {\em IEEE Transactions on Information Theory}, 55(7):3375--3387, July
  2009.
\newblock arXiv:0704.2236.

\bibitem[YHW08]{YHW08}
Dong Yang, Michal Horodecki, and Z.~D. Wang.
\newblock An additive and operational entanglement measure: Conditional
  entanglement of mutual information.
\newblock {\em Physical Review Letters}, 101(14):140501, September 2008.
\newblock arXiv:0804.3683.

\end{thebibliography}

\end{document}